\documentclass[letterpaper,10pt]{amsart}

\usepackage[all]{xy}                        %

\CompileMatrices                            

\UseTips                                    

\input xypic
\usepackage[bookmarks=true]{hyperref}       

\usepackage{amssymb,latexsym,amsmath,amscd}
\usepackage{xspace}
\usepackage{color}
\usepackage[utf8x]{inputenc}

\usepackage{graphicx}
\usepackage{enumerate}

\reversemarginpar

\theoremstyle{plain}
\newtheorem{theorem}{Theorem}[section]
\newtheorem*{theorem*}{Theorem}
\newtheorem{proposition}[theorem]{Proposition}
\newtheorem{corollary}[theorem]{Corollary}
\newtheorem{lemma}[theorem]{Lemma}

\theoremstyle{definition}

\newtheorem{remark}[theorem]{Remark}
\newtheorem{example}[theorem]{Example}

\newcommand{\enm}[1]{\ensuremath{#1}}          %

\newcommand{\cal}[1]{\mathcal{#1}}

\newcommand{\ZZ}{\enm{\mathbb{Z}}}

\newcommand{\PP}{\enm{\mathbb{P}}}

\newcommand{\Ii}{\enm{\cal{I}}}

\newcommand{\Oo}{\enm{\cal{O}}}

\newcommand{\Ol}{\mathcal{O}}

\newcommand{\fq}{\F_q}

\renewcommand{\phi}{\varphi}
\renewcommand{\theta}{\vartheta}
\renewcommand{\epsilon}{\varepsilon}

\renewcommand{\to}[1][]{\xrightarrow{\ #1\ }}

\newcommand{\old}[1]{}

\newcommand{\numberset}{\mathbb}

\newcommand{\K}{\numberset{K}}

\newcommand{\F}{\numberset{F}}

\begin{document}
	
	\title{On the weights of dual codes arising from the GK curve}

\author{Edoardo Ballico}
\address{Universit\`a di Trento, 38123 Povo (TN), Italy}
\curraddr{}
\email{edoardo.ballico@unitn.it}
\thanks{}

\author{Matteo Bonini}
\address{Universit\`a di Trento, 38123 Povo (TN), Italy}
\curraddr{}
\email{matteo.bonini@unitn.it}
\thanks{}

\subjclass[]{}

\keywords{}

\date{}	
	
	\maketitle
	\begin{abstract}
		In this paper we investigate some dual algebraic-geometric codes associated with the Giulietti-Korchm\'aros maximal curve. We compute the minimum distance and the minimum weight codewords of such codes and we investigate the generalized hamming weights of such codes.
	\end{abstract}
	
	{\bf Keywords:} Giulietti-Korchm\'aros curve - algebraic-geometric codes - weight distribution
	
	{\bf MSC Codes:} 14G50 - 11T71 - 94B27
	
	\section{Introduction}
	
	Let $\mathcal{X}$ be an algebraic curve defined over the finite field $\mathbb{F}_q$ of order $q$. 
	We recall that a curve $\mathcal{X}$ is called $\F_q$- maximal if its number of rational points over $\fq$ attains the Hasse-Weil upper bound $$|\mathcal{X}(\F_q)|=q+1+2g(\mathcal{X})q^{1/2},$$ where $g(\mathcal{X})$ is the genus of $\mathcal{X}$.
	
	Since codes with good parameters can be constructed from these curves, many authors studied their properties, see \cite{BMZ2017,BS2019,Goppa81,Goppa82,Hansen1987,MPS2016,Matthews2004}.
	Most of the known examples have been shown to be subcovers of the Hermitian curve $\mathcal{H}$, which is defined over $\F_{q^2}$ by the equation $Y^{q+1}=X^q+X.$
	This led to the question whether every maximal curve is a subcover of the Hermitian curve or not. This question has a negative answer: in \cite{GK2009}, Giulietti and Korchm\'aros introduced an infinity family of curves $\mathcal{C}'$, the so called GK curve, which is maximal over $\mathbb{F}_{q^6}$. 
	
	Codes from the GK curve have been widely investigate, see for example \cite{BB,BMZ2017,CT2016,FG2010}
	
	In most cases, the weight distribution of a given code is hard to be computed. Even the problem of computing codewords of minimum weight can be a difficult  task apart from specific cases. In \cite{BB}, following the approach of \cite{Couv2012}, the authors compute the number of minimum weight codewords of certain dual AG codes arising from the GK curve. For this purpose, they provide a useful algebraic-geometric description for codewords with a given weight which belong to a fixed affine-variety code. These techniques are widely used in literature, see \cite{b,b+,b++,b1}.
	
	In this paper we investigate, using algebraic geometry techniques, codes arising from the GK maximal curve and we give tools to compute the number of minimum weight codewords of such codes. We also investigate some different construction of codes deriving from the ones in \cite{BB} and we study the generalized hamming weights of some codes arising from the GK curve, which are another important pattern for a linear code.

	\section{Preliminaries}
	
	We recall the following results (see \cite[Theorem 1]{b+}).
	
	\begin{lemma}\label{a2}
		Fix integers $r\ge 2$, $m>0$ and $e>0$. Let $Z\subset \PP^r$ be a zero-dimensional scheme such that $\deg (Z)\le 3m+r-3$. If $r\ge 3$ assume that $Z$ spans $\PP^r$.
		We have $h^1(\Ii _Z(m)) \ge e$ if and only if there is $W\subseteq Z$ occurring in this list:
		\begin{itemize}
			\item[(a)] $\deg (W) = m+1+e$ and $W$ is contained in a line;
			\item[(b)] $\deg (W) = 2m+1+e$ and $W$ is contained in a reduced plane conic;
			\item[({c})] $r\ge 3$, $e\ge 2$, and there are an integer $f\in \{1,\dots ,e-1\}$ and lines $L_1, L_2$, such that $L_1\cap
			L_2=\emptyset$, $\deg (L_1\cap Z) = m+1+f$ and $\deg (L_2\cap Z)=m+1+e-f$.
		\end{itemize}
	\end{lemma}

	\begin{lemma}\label{Lemma_cod}
		Let $\F$ be any field and let $\mathbb{P}^n$ denote the projective space of dimension $n$ on $\F$.
		Let $C \subseteq \mathbb{P}^r$ be a smooth plane curve which is a complete intersection. Fix an integer $d>0$, a
		zero-dimensional scheme $E\subseteq C$ and a finite subset $B\subseteq C$ such that $B\cap E_{red}=\emptyset$\footnote{Here $E_{red}$ denotes the reduction of the scheme $E$.}. Denote by $\mathcal{C}$ the code obtained evaluating the vector space $H^0(C,\mathcal{O}_C(d)(-E))$ at the points of $B$.
		Set $c= \deg (C)$,
		$n=|B|$ and assume $|B| > dc -\deg (E)$. The following facts hold.
		\begin{enumerate}
			\item The code $\mathcal {C}^{\bot}$ has length $n$ and dimension $k= h^0(C,\mathcal
			{O}_C(d))-\deg (E) +h^1(\mathbb {P}^n,\mathcal {I}_E(d))$.
			
			\item The minimum distance
			of $\mathcal {C}^{\bot}$ is the minimal cardinality, say $z$, of a subset of $S
			\subseteq B$
			such that
			$$h^1(\mathbb {P}^n,\mathcal {I}_{S\cup E}(d)) >h^1(\mathbb {P}^n,\mathcal
			{I}_E(d)).$$
			
			\item A codeword of $\mathcal{C}^{\bot}$
			has weight $z$ if and only if it is supported by a subset $S\subseteq B$ such
			that
			\begin{enumerate}
				\item $|S| = z$,
				\item $h^1(\mathbb {P}^n,\mathcal {I}_{E\cup S}(d)) >h^1(\mathbb
				{P}^n,\mathcal {I}_E(d))$,
				\item $h^1(\mathbb {P}^n,\mathcal {I}_{E\cup S}(d))
				>h^1(\mathbb {P}^n,\mathcal {I}_{E\cup S'}(d))$ for any $S'\varsubsetneq S$.
			\end{enumerate}
		\end{enumerate}
	\end{lemma}

	A zero-dimensional scheme $Z\subset \mathbb {P}^r$ is said to be {\it curvilinear} if at each $P\in Z_{red}$ the Zariski tangent space of $Z$ has dimension $\le 1$. A zero-dimensional scheme is contained in a smooth curve (easy). A zero-dimensional scheme is curvilinear if and only if it has finitely many subschemes (for the ``~only if~'' part use that it is contained in a smooth curve,
	for the ``~if~'' part use that a non-curvilinear subscheme has infinitely many
	subschemes with degree $2$). In this note we point out the following partial extension of \cite{b}, Theorem 1,
	to the case of non-reduced, but curvilinear subschemes.
		
	We recall the following results (\cite[Theorem 1]{b++}).
		
	\begin{theorem}\label{z1}
		Fix an integer $m \ge 3$. Let $Z\subset \mathbb {P}^r$, $r\ge 3$, be a curvilinear zero-dimensional scheme spanning $\mathbb {P}^r$. If $r=3$, then assume $\deg (Z)< 3m$.
		If $r\ge 4$, then assume $\deg (Z) \le 4m+r-5$ and $\deg (Z\cap M)<3m$ for all $3$-dimensional
		linear subspaces $M\subset \mathbb {P}^r$. We have $h^1(\mathcal {I}_Z(m)) >0$ if and only if either
		there is a line $D$ with $\deg (D\cap Z)\ge m+2$ or there is a conic $D'$ with $\deg (D'\cap Z)\ge 2m+2$.
	\end{theorem}

	With minimal modifications of the proof of \cite{b} we get the following result	
	
	\begin{lemma}\label{x1}
		Let $Y\subset \PP^3$ be a smooth and connected projective curve defined over an algebraically closed field. Fix a zero-dimensional scheme $A\subset Y$
		and a finite set $B \subset Y$ such that $A\cap B =\emptyset$. Set $Z:= A\cup B$. Assume $\deg(A) $
		Assume $\deg (A)<3m$, $\deg (Z) \le 4m+2$ and $\deg (Z\cap H)\le 4m-5$ for each plane $H\subseteq \PP^r$
		We have $h^1(\mathcal {I}_Z(m)) >0$ if and only if there is $W\subseteq Z$ as in one of the following cases:
		\begin{itemize}	
			\item[(a)] $\deg (W) =m+2$ and $W$ is contained in a line;
			\item[(b)] $\deg (W)=2m+2$ and $W$ is contained in a plane conic;
			\item[(c)] $\deg(W) = 3m$ and $W$ is the complete intersection of a degree $3$ plane curve and a degree $m$ surface;
			\item[(d)] $\deg (W) \ge 3m+1$ and $W$ is contained in a degree $3$ plane curve;
			\item[(e)] $\deg (W) =3m+2$ and $W$ is contained in a reduced and connected degree $3$ curve spanning $\mathbb {P}^3$.
		\end{itemize}
	\end{lemma}

	\section{GK curve}
	
	Denote by $PG(3,q^6)$ the three dimensional projective space over the field $\mathbb{F}_{q^6}$ with $q^6$ element. The Giulietti-Korchm\'aros curve $\mathcal{GK}$  is a non-singular curve in $PG(3,q^6)$ defined by the affine equations
	\begin{equation}
	\label{Eq:GK}
	\begin{cases}
	Z^{q^2-q+1}=Y^{q^2}-Y\\
	Y^{q+1}=X^q+X
	\end{cases}.
	\end{equation}
	
	Arbitrary complete intersections in $\PP^r$ are defined and studied in \cite[Ex. II.8.4 and III.5.5]{h}. We always consider the case of smooth space curves, complete intersection of a surface $S$ of degree $a$ and a surface of degree $b\ge a$ (for GK we have $a=q+1$ and $b=q^2$. We have $h^1(\PP^3,\Ii _C(t)) =0$
	for all $t\in \ZZ$ and hence (for a smooth curve) and any zero-dimensional scheme $Z\subset C$) we have $h^1(C,\Oo_C(t)(-Z)) = h^1(\Oo _C(t)) +h^1(\PP^3,\Ii _Z(t))$. We have
	the exact sequences
	\begin{equation}\label{eqa1}
	0 \to \Oo _{\PP^3}(t-a) \to \Oo _{\PP^3}(t)\to \Oo _S(t)\to 0
	\end{equation}
	\begin{equation}\label{eqa2}
	0\to \Oo _S(t-b) \to \Oo _S(t)\to \Oo _C(t)\to 0
	\end{equation}
	We have $h^0(\Oo _{\PP^3}(t)) =\binom{t+3}{3}$ for all $t\ge 0$. From (\ref{eqa1}) and (\ref{eqa2}) we get $h^0(\Oo _C(t)) =\binom{t+3}{3}$ for all $t<a$ and $h^0(\Oo _C(t)) =\binom{t+3}{3} -\binom{t-a+3}{3}$ if $a\le t<b$ (see below the proof of the case $t \ge b$). From (\ref{eqa1}) and the fact that $h^i(\Oo _{\PP^3}(x)) =0$ for $i=1,2$ and all $x\in \ZZ$ we get $h^1(\Oo _S(x)) =0$ for all $x\in \ZZ$ and $h^0(\Oo _C(t)) = h^0(\Oo _S(t))$ if $t<b$ and $h^0(\Oo _C(t)) =h^0(\Oo _S(t)) -h^0(\Oo _S(t-b))$ for all $t\ge b$.
	Thus for all $t\ge b$ we have $h^0(\Oo _C(t)) = \binom{t+3}{3} -\binom{t-a+3}{3} +\binom{t-b}{3} -\epsilon$, where $\epsilon =0$ if $t<b+a$ and $\epsilon = \binom{t-b-a}{3}$ if $t\ge b+a$.

	\begin{proposition}
		Let $L$ be a tangent to $\mathcal{GK}$ at a point $P$. Then $I(L,\mathcal{GK},P)=q^2-q+1$ or $I(L,\mathcal{GK},P)=q+1$.
	\end{proposition}
	
	\proof
	
	We know that the tangent to an affine point of this curve $P=(x_0,y_0,z_0)$ has equation
	\[
	\begin{cases}
	(Y-y_0)+z_0^{q^2-q}(Z-z_0)=0\\
	-(X-x_0)+y_0^q(Y-y_0)=0
	\end{cases}
	\]
	The parametric equation of this line is, for $z_0\ne0$
	\[
	\begin{cases}
	X=x_0+y_0^qt-y_0^{q+1}\\
	Y=t\\
	Z=\frac{-t+y_0+z_0^{q^2-q+1}}{z_0^{q^2-q}}=\frac{t+y_0^{q^2}}{z_0^{q^2-q}}
	\end{cases}
	\]
	
	while for $z_0=0$ it is
	\[
	\begin{cases}
	X=x_0\\
	Y=y_0\\
	Z=t
	\end{cases}
	\]
	and the solution corresponding to $P$ is $t=y_0$ and $t=z_0$ respectively.
	
	Suppose $z_0\ne0$.
	
	Substituting the equation of the affine equation of the GK curve gives us
	\begin{equation}
	\label{eq.1GK}
	\begin{cases}
	\left(\frac{-t+y_0^{q^2}}{z_0^{q^2-q}}\right)^{q^2-q+1}-t^{q^2}+t=0\\
	-t^{q+1}+(x_0+y_0^qt-y_0^{q+1})^q+x_0+y_0^qt-y_0^{q+1}=0
	\end{cases}
	\end{equation}
	
	The first equation becomes
	\[
	\begin{split}
	0&=(-t+y_0^{q^2})^{q^2-q+1}-(t^{q^2}-t)z_0^{(q^2-q)(q^2-q+1)}\\
	&=(y_0^{q^2}-t)^{q^2-q+1}-(t^{q^2}-t)(y_0^{q^2}-y_0)^{q^2-q}\\
	\end{split}
	\]
	which has $t=y_0$ as a root, its derivative is 
	\[
	-(y_0^{q^2}-t)^{q^2-q}+(y_0^{q^2}-y_0)^{q^2-q}=(-(y_0^{q^2}-t)^{q-1}+(y_0^{q^2}-y_0)^{q-1})^q
	\]
	and since $t=y_0$ is a root of $-(y_0^{q^2}-t)^{q-1}+(y_0^{q^2}-y_0)^{q-1}$ we have that $t=y_0$ is a root of (\ref{eq.1GK}) with multiplicity at least $q+1$.
	
	By direct computations the second equation becomes
	\[
	\begin{split}
	0=&-t^{q+1}+x_0^q+y_0^{q^2}t^q-y_0^{q^2+q}+x_0+y_0^qt-y_0^{q+1}=-t^{q+1}+y_0^{q^2}t^q-y_0^{q^2+q}+y_0^qt\\
	=&t^q(-t+y_0^{q^2})-y_0^{q}(-t+y_0^{q^2})=(-t+y_0^{q^2})(t-y_0)^q
	\end{split}
	\]
	and from this we get that $t=y_0$ is a root with multiplicity $q+1$ if $y_0\in\F_{q^2}$ or $q$ is $y_0\not\in\F_{q^2}$.

	Now we deal with the remaining case $z_0=0$.
	Substituting the equation of the affine equation of the GK curve gives us
	\[
	\begin{cases}
	t^{q^2-q+1}=y_0^{q^2}-y_0\\
	y_0^{q+1}=x_0^q+x_0
	\end{cases}
	\]
	where the second is not an equation in $t$ but just a compatibility condition. So, if this holds we get that 
	\[
	t^{q^2-q+1}=0
	\]
	In this case the tangent in $P$ is a $q^2-q+1$-secant.
	
	The last case we have to study is the case $P_\infty=(1:0:0:0)$, the homogenized equations of the curve are
	\begin{equation}
	\label{Eq:GKproiett}
	\begin{cases}
	Z^{q^2-q+1}T^{q-1}=Y^{q^2}-YT^{q^2-1}\\
	Y^{q+1}=X^qT+XT^{q}
	\end{cases}.
	\end{equation}
	the equation of the tangent line will be then
	\[
	\begin{cases}
	X=1\\
	Y=0
	\end{cases}
	\]
	and the multiplicity intersection at this point with the tangent is $q^2-q+1$.
	\endproof

	\section{Codes from the GK curve}
	We recall results for the intersections of algebraic curves and $\mathcal{GK}$ and for the minimum distance of the one point AG code $C(D,G_m)^{\bot}$, where $G_m=m(q^3+1)P_{\infty}$, $P_{\infty}=(1:0:0:0)$, and $D=\sum_{P \in \mathcal{GK}(\mathbb{F}_{q^6} )\setminus \{P_{\infty}\}} P$, see \cite{BB} for details.
	
	\begin{proposition}\label{Prop:lines}
		Let $r \subset PG(3,q^6)$ be a line. 	Then 
		$$|r \cap \mathcal{GK}|\leq q^2-q+1.$$
		Also, any $(q^2-q+1)$-secant is parallel to the $z$-axis and all the $(q^2-q+1)$ common points are not $\F_{q^2}$-rational.
	\end{proposition}
	\begin{proposition}\label{prop:NumberOfLines}
		The total number of $(q^2-q+1)$-secants of the $\mathcal{GK}$ is $(q+1)(q^5-q^3)$.
	\end{proposition}
	Remember that each point lies in exactly one of such secants.
	\begin{proposition}
		Let $\mathcal{X}$ be a curve of degree $\alpha\le q$ in $PG(3,q^6)$. Then the size $|\mathcal{X}\cap \mathcal{GK}(\mathbb{F}_{q^6})|$ is at most  
		$$\left\{
		\begin{array}{ll}
		\alpha(q^2-q+1),& \textrm{ if } \mathcal{X} \textrm{ is reducible},\\
		\alpha(q+1), & \textrm{ if } \mathcal{X} \textrm{ is absolutely irreducible}.\\
		\end{array}
		\right.$$
	\end{proposition}
	\begin{proposition}\label{Prop:MinimumDistance}
		Let $d^*\leq d$ be the designed Goppa minimum distance of $C(D,G_m)^{\bot}$, $m\geq2$. Then 
		\begin{enumerate}
			\item $d=m+2$ when $m\le q^2-q-1$;
			\item $d=2m+2$ when $m=q^2-q$;
			\item $d=3m$ when $m=q^2-q+1$;
			\item $d\ge3m+1$ when $q^2-q+1<m\le q^2-1$;
			\item $d\ge d^*$ when $m>q^2-1$.
		\end{enumerate} 
	\end{proposition}	
	
	\subsection{The family $\overline{C}_S$}
	
	Condider now a set $S\subset\mathcal{GK}(\F_{q^6})$ and the corresponding divisor \[D_S=D-\sum_{P\in S} P\] and call 
	\[
	S_1=\{P\in S\,:\,P\in\mathcal{GK} (\F_{q^2}) \},\quad S_2=S\setminus S_1.
	\]
	The following result comes from a straight application of \ref{prop:NumberOfLines}.
	\begin{proposition}
		Let $q+1\le m \le 2(q+1)$ and $D_S$ defined as before. Consider the code $\overline{C}_S=C(D_S, (q^3+1)mP_\infty)^{\bot}$, if
		\[
		|S_2|< (q^2-q+1-m)(q+1)(q^5-q)
		\]
		then $\overline{C}_S$ is a $[n-|S|,\ell(D_S)-\ell(D_S-G_m),m+2]_q$-code. Moreover, if $S=S_1$ then the number of minimum weight codewords of $\overline{C}_S$ is given by
		\[
		A_{m+2}(\overline{C}_S)=(\ell+1)(\ell^5-\ell^3)(\ell^6-1)\binom{\ell^2-\ell+1}{m+2}.
		\]
	\end{proposition}
	\proof
	Following the proof of Proposition (\ref{Prop:MinimumDistance}) and noticing that if $|S_2|< (q^2-q+1-m)(q+1)(q^5-q)$ there is at least a ($m+2$)-secant line the result holds.
	\endproof
	\begin{remark}
		Actually this bound can be improved depending on the composition of $S_2$, i.e. if the points in $S_2$ are chosen in a way such that at least $m+2$ of them lie in the intersection between one line and $\mathcal{GK}$, then all the other ones can be taken leaving the minimum distance unchanged (while decreasing the dimension of the code, so improving the code itself).
	\end{remark}

\subsection{Three-point codes}
\begin{theorem}\label{u5}
	Fix any three distinct points $P_1,P_2,P_3\in 
	\mathcal(GK)(\mathbb {F}_{q^6})\setminus\{P_\infty\}$ and assume $P_1$, $P_2$ and $P_3$ to span $\mathbb{P}^2$ and be such that their connecting line is not parallel to the $z$ axis.
	Set $B:=  \mathcal(GK)(\mathbb {F}_{q^6})\setminus \{P_1,P_2,P_3\}$. Fix
	an integer $d \ge 5$ such that $1\le d\le q-1$ and integers $a_1,a_2,a_3 \in
	\{1,\dots ,d\}$
	such that $a_1+a_2+a_3 \le 3d-5$ and $a_i=d$ for at most one index $i \in \{
	1,2,3\}$. Set $E:= a_1P_1+a_2P_2+a_3P_3$. 
	Let $\mathcal {C}:= \mathcal {C}(B,d,-E)$ be the code obtained evaluating the
	vector space $H^0(\mathcal(GK),\Ol_\mathcal(GK)(d)(-E))$ on the set $B$. Then $\mathcal {C}$ is a code
	of length $n=|B|=q^8-q^6+q^5-2$ and dimension $k=\binom{d+3}{3}-a_1-a_2-a_3$. For
	any $i \in \{ 1,2,3 \}$ let $L_i$ denote the line spanned
	by $P_j$ and $P_h$ with $\{i,j,h\} = \{1,2,3\}$. Then $\mathcal {C}^{\bot}$ has	minimum distance $d$ and its minimum-weight codewords are exactly the ones whose support is formed by	$d$ points of $B\cap L_i$
	for some $i \in \{ 1,2,3\}$.
\end{theorem}
\begin{proof}

	The length of $\mathcal{C}$ is obviously $n = q^8-q^6+q^5-2$. From what we said previously we have $h^0(\mathcal(GK),\mathcal {O}_\mathcal(GK)(d)) =\binom{d+3}{3}$. If, say, $a_1\ge a_2\ge a_3$, the assumptions $a_1\le d$ and $a_1+a_2+a_3   \le 3d-1$
	give $a_i\le d+2-i$ for all $i$. Hence our previous computations tell us that $h^1(\mathbb
	{P}^2,\mathcal {I}_E(d)) = 0$ and so $ h^0(\mathcal(GK),\mathcal {O}_\mathcal(GK)(d)(-E)) =
	\binom{d+3}{3} -a_1-a_2-a_3 =k$. 
	
	Since $|B|> d\cdot \deg (\mathcal(GK))$, there is not a non-zero element
	of $H^0(\mathcal(GK),\mathcal {O}_\mathcal(GK)(d))$ vanishes at all the points of $B$. Hence $\mathcal{C}$ has dimension $k$.
	By Lemma \ref{Lemma_cod} it is sufficient to prove the following two facts. 
	\begin{itemize}
		\item[(a)] $h^1(\mathbb {P}^3,\mathcal {I}_{E\cup A}(d)) =0$ for all
		$A\subseteq B$ such that $|A|\le d-1$.
		\item[(b)] For any
		$S\subseteq B$ such that $|S| =d$ we have
		$h^1(\mathbb {P}^3,\mathcal {I}_{E\cup S}(d)) >0$ if and only if $S\subseteq
		L_i$ for some $i \in \{ 1,2,3\}$.
	\end{itemize}
	Each line $L_i$ contains at most $q-1$ points of $B$ while $\deg (E\cap L_i)=2$.
	Hence for any $S\subseteq L_i\cap B$ with $|S| = d$ we have $h^1(\mathbb
	{P}^2,\mathcal {I}_{E\cup S}(d)) >0$ from Lemma \ref{Lemma_cod}. 
	
	Let $E_i:=a_iP_i$, clearly $E=E_1+ E_2+ E_3$ (seen as a divisor). 
	
	Fix a set $S\subseteq B$ such that $|S|\le d$ and assume $h^1(\mathbb {P}^3,\mathcal {I}_{E\cup S}(d)) >0$. We have $S\cap
	\{P_1,P_2,P_3\} =\emptyset$ and $\deg (E\cup S) = a_1+a_2+a_3+|S|$.
	Since $a_1+a_2+a_3 +|S| \le 4d-5$, we may apply
	Proposition \ref{x1} to the scheme $E\cup S$. 
	
	Let $T\subseteq \mathbb {P}^n$ be the curve arising from the statement of the lemma. Set $x:= \deg (T) \in
	\{1,2,3\}$ and $e_i:= \deg (T\cap E_i)$ for $i \in \{1,2,3\}$. We have $0 \le	e_i\le a_i$. 
	
	If $e_i\ge x+1$ then we have that the tangent at $P_i$ is $L_{\mathcal(GK),P_i}\subseteq T$. Assume $e_i\le x$ for all $i \in
	\{ 1,2,3\}$. 
	For $x=2$ we get $\deg (T\cap (E\cup S)) \le 2d+1$.
	For $x=3$ we get $\deg (T\cap (E\cup S)) \le 3d-1$. Finally, for $x=1$
	we may have
	$e_i>0$ only for at most two indices, say $i=1,2$. Since $|S|\le d$, we get
	$|S|+e_1+e_2\ge d+2$ and $|S|+e_1+e_2=d+2$ if and only if $T =L_3$, $S\subseteq	L_3\cap B$ and $|S|=d$.

	Now assume that $T$ contains one of the lines $L_{\mathcal(GK),P_i}$, say $L_{\mathcal(GK),P_1}$. Let
	$T'$ be the curve whose
	equation is obtained dividing an equation of $T$ by an equation of $L_{\mathcal(GK),P_1}$.
	We have
	$\deg (T') =x-1$, $T'+L_{\mathcal(GK),P_1} = T$ (as divisors of $\mathbb {P}^2$) and $T =
	L_{\mathcal(GK),P_1}\cup T'$
	(as sets). Since $L_{\mathcal(GK),P_1}\cap B=\emptyset$, we have $T\cap S
	= T'\cap S$ and $\deg (T\cap (E\cup S)) = \deg (T'\cap (E_2\cup E_3\cup S))$.
	\begin{itemize}
		
		\item[(i)] If $x=1$, we
		get $T\cap S=\emptyset$ and $\deg (T\cap E)
		= a_1 \le d$, a contradiction.
		
		\item[(ii)] Assume $x=2$.  The curve $T'$ must be a line such that $\deg (T'\cap
		(E_2\cup E_3 \cup S))
		\ge 2d+2-a_1$. If either $T' = L_{\mathcal(GK),P_2}$, or $T' = L_{\mathcal(GK),P_3}$, we
		get $T'\cap S = \emptyset$ and $\deg (T'\cap (E_2\cup E_3\cup S)) \le \max
		\{e_2,e_3\} \le d$, a contradiction. If neither $T' = L_{\mathcal(GK),P_2}$, nor $T' =
		L_{\mathcal(GK),P_3}$, then $\deg (T'\cap E_2)\le 1$,
		$\deg (T\cap E_3)\le 1$ and $\deg (T'\cap (E_2\cup E_3)) =2$ if and only if $T'
		=L_1$.
		Since $|S| \le d$ we deduce $\deg (T\cap (E\cup S))  \le a_1+2+|S|$. Moreover,
		the equality holds
		if and only if $T' = L_1$ and $S\subseteq L_1$. Since
		$\deg (T\cap (E\cup S)) \ge 2d+2$ by assumption, $|S|=d$ and $S\subseteq L_1$,
		as claimed.
		
		\item[(iii)]Now
		assume $x=3$. We get $\deg (T'\cap (E_2\cup E_3\cup S)) \ge 3d-a_1$ and $T'$ is
		a conic.
		If neither
		$L_{\mathcal(GK),P_2}$, nor $L_{\mathcal(GK),P_3}$, is a component of $T$
		then $e_2\le 2$ and $e_3\le 2$ and so $|T'\cap S| \ge
		3d-4-a_1
		\ge 2d-4 > d$. If, say, $T'$ contains
		$L_{\mathcal(GK),P_2}$ and $T''$ is the line with $T' = T''+L_{\mathcal(GK),P_2}$, then we get
		$|(S\cup E_3)\cap T''| \ge 3d-a_1-a_2$. Since
		$a_1+a_2 \le 2d-1$ we deduce $\deg (T''\cap (E_3\cup S)) \ge d+1$. Since $\deg
		(T''\cap E_3)\le 1$,
		we get
		$a_1+a_2=2d-1$, say $a_1=d$ , $a_2=d-1$ and that $S$ is formed by $d$ points on
		a line
		$T''$ through $P_3$. If either $T''= L_1$ or $T''= L_3$, then we are done. In
		any case it is sufficient to prove that $E_1\cup E_2\cup
		\{P_3\}\cup S$ is not
		the complete intersection of $T = L_{\mathcal(GK),P_1}\cup L_{\mathcal(GK),P_2}\cup T''$ and a degree
		$d$ curve, say $C_d$.
		Since $a_2=d-1$, $E_2$ is not the complete intersection of $L_{\mathcal(GK),P_2}$ and
		$C_d$, while $L_{\mathcal(GK),P_2}\cap (\{P_3\}\cup S)
		= \emptyset$, a contradiction.
	\end{itemize}
\end{proof}

\begin{corollary}
	Using the notation of the previous theorem, the number of minimum weight codewords of the code above is given by
	\[
	A_d(\mathcal{C})=(q^6-1)\sum_{i=1}^3 \binom{|L_i\cap \mathcal{GK}|}{d}\le (q^6-1)3\binom{p+1}{d}.
	\]
	where the binomial coefficient is meant to be zero if $d>|L_i\cap \mathcal{GK}|$ for some $i$.
\end{corollary}

\begin{example}
	Let $q=7$ and consider the affine equation of $\mathcal{GK}$ over the field $\F_{q^6}$
	\begin{equation*}
	\begin{cases}
	Z^{43}=Y^{49}-Y\\
	Y^{8}=X^7+X
	\end{cases}.
	\end{equation*}
	Consider $P_1=(0:0:0:1)$, $P_2=(1:3:0:1)$ and $P_3=(1:4:0:1)$. The three points are in general position and their connecting line are not parallel to the $Z$ axis, so the conditions of the previous theorem are satisfied. Moreover, by direct computations, the three lines $L_1$, $L_2$ and $L_3$ are $8$-secants of $\mathcal{GK}$. Consider now $d=6$ and $a_1=6$, $a_2=a_3=3$ and call $\mathcal{C}=\mathcal{C}(B,6,6P_1+3P_2+3P_3)$.  
	From Theorem \ref{u5} we have that the minimum distance of $\mathcal{C}$ is $d=6$ and the minimum weight codewords are exactly
	\[
	A_6(\mathcal{C})=(7^6-1)3\binom{8}{6}=(7^6-1)84
	\]
	
\end{example}

\section{Generalized Hamming Weights of codes arising from the GK curve}

Let $\K=\F_q$ a finite field with $q$ elements. 
Let $C\subset \K^n$ be a linear $[n,k]$ code over $\K$. 
We recall that the \textit{support} of $C$ is defined as follows
$$
supp(C)=\{i \mid c_i \ne 0 \mbox{ for some } c \in C\}.
$$
So $\sharp supp(C)$ is the number of nonzero columns in a generator matrix for $C$.
Moreover, for any $1\leq v\leq k$, the \textit{$v$-th generalized Hamming weight} of $C$
\[
d_v(C)=\min\{\sharp supp(\mathcal D) \mid \mathcal D \mbox{ is a linear subcode of $C$ with } dim(\mathcal D)=v\}.
\]
In other words, for any integer $1\leq v\le k$, $d_v(C)$ is the $v$-th minimum support weights, i.e. the minimal integer $t$ such that
there are an $[n,v]$ subcode $\mathcal {D}$ of $C$ and a subset $S\subset \{1,\dots ,n\}$ such that $\sharp (S)=t$ and each codeword of $\mathcal {D}$ has zero coordinates outside $S$. The sequence $d_1(C),\ldots,d_k(C)$ of generalized Hamming weights (also called \textit{weight hierarchy} of $C$) is strictly increasing (see Theorem~7.10.1 of \cite{hp}).
Note that $d_1(C)$ is the minimum distance of the code $C$.

\begin{lemma}
	\label{lemma1}
	Let $S\subset B$ be the support of a codeword of $C^\bot$. Assume that there exists a surface  $T\subset \mathbb{P}^3$ such that $h^1(\mathbb{P}^3,\mathcal{I}_{Res_T(E\cup S)}(d-k))=0$, where $k=\deg(T)$. Then $S\subset T$.
\end{lemma}

\proof
Let $W$ (reps. $W^\prime)$ be the subcode of $C^\bot$ formed by the codewords whose support is contained in $S$ (resp. $S\cap T$). Clearly $W^\prime\subseteq W$. From Proposition \ref{a2} we get $h^1(\mathbb{P}^3,\mathcal{I}_{E\cup S}(d))=h^1(\mathbb{P}^3,\mathcal{I}_{T\cap (E\cup S)}(d))$. From this we obtain $W=W^\prime$, which means that the thesis is proved.
\endproof

\begin{theorem}
	Fix a positive integer $d\le \deg(\mathcal{GK})-1$, a zero dimensional scheme $E\subseteq \mathcal{X}$ defined over $\K$ and a set $B\subseteq \mathcal{GK}(\K)\setminus E_{red}$ such that $\deg(E)\le d+1$ and set $C:=C(B,\mathcal{O}_{\mathcal{GK}(d)(-E)})$. Assuming that each line is such that $\deg(L\cap(E\cup B))\le d+1$ and that there exists a conic such that
	\begin{enumerate}[(i)]
		\item $\deg(D\cap E)+|B\cap D|\ge 2d+2$;
		\item for each conic $\mathcal{C}$ such that $T\ne D$ we have $\deg(T\cap(E\cup B))\le 2d+1$.
	\end{enumerate}
	For any integer $s$ such that $2d+2-\deg(D\cap E)\le s\le |B\cap D|$ and for each integer $h$ with $1\le h\le \min\{|B\cap D|-2d-2+\deg(D\cap E),d-2-\deg(E)\}$ each $h$-dimensional linear subspace oc $C^\bot$ computing $d_h(C^\bot)$ is supported by some $S\subset\Sigma(2d+h-1+\deg(D\cap E))$ and each element of $\Sigma(d+h-1+\deg(D\cap E))$ is in the support of a $h$-dimensional linear subspace.
\end{theorem}

\proof
Fix an integer $e\ge1$ and any $S\subseteq B$. Lemma \ref{Lemma_cod} tell us that $S$ contains the support of an $e$-dimensional subspace of $C^\bot$ if and only if $h^1(\mathbb{P}^3,\mathcal{I}_{E\cup S}(d))\ge e$. Fix $S\subseteq B$ such that it is the support of a codeword of $C^\bot$ with weight $\le 3d+1-\deg(E)$, hence $\deg(E\cup S)\le 3d-1$ and Lemma \ref{a2} tells us the value of $h^1(\mathbb{P}^2,\mathcal{I}_{E\cup S}(d))$. Let $t$ be the minimal integer such that $h^1(\mathbb{P}^3,\mathcal{I}_{E\cup S}(d))\ge0$, clearly $t\le d$. If $\deg(T\sup S)\le 2d+1$ then  there is a line $L$  with $\deg(L\cap(E\cup B))\ge\deg(L\cap(E\cup S))\ge d+2$, but this cannot happen since we excluded this possibility. For this reason we can assume that $\deg(E\cup S)\ge2d+2$. Since $d\ge4$, we have that there is a line $L$ such that $\deg(L\cap(E\cup S))\ge d$ scheme-theoretic base locus of the set $|\mathcal {I}_D(2))|$ of all quadric surfaces containing $D$. Take any quadric $T\supset D$. Since $\deg (Res _T(E\cup B))\le 3d-1-(2d+2) \le d-1$, we have
$h^1(\mathbb{P}^3,\mathcal{I}_{Res_T(E\cup B)}(d-2))=0$. Thus Lemma \ref{lemma1} gives $E\cup B\subset T$. Since $D$ is scheme-theoretically the base locus of $|\mathcal {I}_D(2))|$, we get $E\cup B\subset D$.
\endproof

\section*{Acknowledgments}
The research of M. Bonini was supported by the Italian National Group for Algebraic and Geometric Structures and their Applications (GNSAGA - INdAM).

\bigskip

	\bibliographystyle{amsplain}
	\providecommand{\bysame}{\leavevmode\hbox to3em{\hrulefill}\thinspace}
	\providecommand{\MR}{\relax\ifhmode\unskip\space\fi MR }
	\providecommand{\MRhref}[2]{%
		\href{http://www.ams.org/mathscinet-getitem?mr=#1}{#2}
	}
	\providecommand{\href}[2]{#2}

\end{document}